\documentclass[11pt]{article}
\usepackage{graphicx}
\usepackage{longtable,lscape}
\usepackage{amsthm}
\usepackage{amsfonts}
\usepackage{amsmath}
\usepackage{bbm}
\usepackage{float}
\usepackage{url}
\newtheorem{definition}{Definition}
\newtheorem{theorem}{Theorem}

\newcommand{\captionfonts}{\footnotesize}
\makeatletter  
\long\def\@makecaption#1#2{%
  \vskip\abovecaptionskip
  \sbox\@tempboxa{{\captionfonts #1: #2}}%
  \ifdim \wd\@tempboxa >\hsize
    {\captionfonts #1: #2\par}
  \else
    \hbox to\hsize{\hfil\box\@tempboxa\hfil}%
  \fi
  \vskip\belowcaptionskip}
\makeatother 
\begin{document}
\title{Entanglement of Conceptual Entities in \\ Quantum Model Theory (QMod)}
\author{Diederik Aerts and Sandro Sozzo \vspace{0.5 cm} \\ 
        \normalsize\itshape
        Center Leo Apostel for Interdisciplinary Studies \\
        \normalsize\itshape
        Brussels Free University \\ 
        \normalsize\itshape
         Krijgskundestraat 33, 1160 Brussels, Belgium \\
        \normalsize
        E-Mails: \url{diraerts@vub.ac.be,ssozzo@vub.ac.be} \\
               }
\date{}
\maketitle
\begin{abstract}
\noindent
We have recently elaborated \emph{Quantum Model Theory} (\emph{QMod}) to model situations where the quantum effects of contextuality, interference, superposition, entanglement and emergence, appear without the entities giving rise to these situations having necessarily to be of microscopic nature. We have shown that QMod models without introducing linearity for the set of the states. In this paper we prove that QMod, although not using linearity for the state space, provides a method of identification for entangled states and an intuitive explanation for their occurrence. We illustrate this method for entanglement identification with concrete examples.
\end{abstract}


\vspace*{-0.2cm}
\section{Introduction\label{intro}}
We have recently presented \emph{Quantum Model Theory} (\emph{QMod}) \cite{aertssozzo2012a}, a modeling theory worked out to describe situations entailing effects, such as, \emph{interference}, \emph{contextuality}, \emph{emergence} and \emph{entanglement}, which are typical of the micro-world but also occur at macroscopic level and even outside physics \cite{aerts2009,aerts2009b,aerts2010,aerts2010b}. QMod rests on a generalization of the standard Hilbert space quantum formalism, namely the \emph{State Context Property} (\emph{SCoP}) formalism \cite{aerts2002}, developed in Brussels when investigating the structure of concepts, and how they combine to form sentences and texts \cite{aertsgabora2002,aertsgabora2005a,aertsgabora2005b}. The SCoP formalism was further used to analyze aspects of concepts and inspired contextual approaches \cite{aertsgabora2002,gabora2007,nelson2007,gaboraroschaerts2008,flenderkittobruza2009,gaboraaerts2009,dhooghe2010,aertsczachorsozzo2010,velozgaboraeyjolfsonaerts2011,aertssozzo2011}. However, the SCoP formalism is very general, hence QMod has been developed to be a formalism closer to the complex Hilbert space of standard quantum theory but, at the same time, general enough to cope with the modeling of the main quantum effects identified in the domains different from the micro-world.

QMod makes it possible to describe not only concepts and their combinations, but any kind of entity in which the above quantum effects play a relevant role. Furthermore, it is a generalization of classical and quantum theory in a very similar way to how the relativistic manifold formalism is a generalization of special relativity and of Newtonian physics in space time.

In this paper we focus on entanglement and emergence, and show that these effect find a very natural description in QMod. We first introduce in Sec. \ref{sectionrepresentation} a representation theorem which shows how one can construct a real or complex representation for a general entity. Then, we apply this theorem to model two specific examples in Sec. \ref{entanglement}. In the first example, we consider the concept \emph{The Animal Acts}, which is a combination of the concepts \emph{Animal} and \emph{Acts}. By using the experimental data collected in \cite{aertssozzo2011} we analyze the entanglement between these two concepts (\ref{ent_conc}). Finally, we consider the entity {\it Vessel of Water}, and show that states of this entity can be prepared which are not product states, i.e. they are entangled (Sec. \ref{ent_vessels}).

\vspace*{-0.2cm}
\section{A representation theorem\label{sectionrepresentation}}
In this section we resume the essentials of the representation theorem proved in detail in \cite{aertssozzo2012a} that are needed to attain our results in the following sections. Let us begin with the abstract description of an entity in QMod. An entity is a collection of aspects of reality that hang together in such a way that different states exist without loosing the possibility of identification of the same entity in each of these states. Sometimes only one state exists, this is then the limiting case, and the entity is then just a situation.

\begin{definition}\label{def}
We consider an entity $S$ that can be in different states, and denote states by $p, q, \ldots$, and the set of states by $\Sigma$. Different measurements can be performed on the entity $S$ being in one of its states, and we denote measurements by $e, f, \ldots$, and the set of measurements by ${\cal M}$. With a measurement $e \in {\cal M}$ and the entity in state $p$, corresponds a set of possible outcomes $\{x_1, x_2, \ldots, x_j, \ldots, x_n\}$, and a set of probabilities $\{\mu(x_j,e,p)\}$, where $\mu(x_j,e,p)$ is the limit of the relative frequency of the outcome $x_j$, the situation being repeated where measurement $e$ is executed and the entity $S$ is in state $p$. We denote the final state corresponding to the outcome $x_j$ by means of $p_j$.
\end{definition}
Let us now come to the representation theorem. It states that it is always possible to realize the situation in Def. \ref{def} by means of a specific mathematical structure using a space of real numbers where the probabilities are derived as Lebesgue measures of subsets of real numbers. Moreover, a complex number realization exists as well, where the probabilities are calculated by making use of a scalar product similar to the one used in the quantum formalism \cite{aertssozzo2012a}.

\begin{theorem} \label{theoremrepresentation}
Consider a measurement $e \in {\cal M}$ and a state $p \in \Sigma$, and the set of probabilities $\{\mu(x_j, e, p)\}$, where $\{x_1, \ldots, x_j, \ldots, x_n\}$ is the set of possible outcomes given $e$ and $p$, then it is possible to work out a representation of this situation in $\mathbb{R}^n$ where the probabilities are given by Lebesgue measures of appropriately defined subsets of $\mathbb{R}^n$, and a representation in $\mathbb{C}^m$ where the measurement is modeled within the mathematical formalism of standard quantum theory defined on $\mathbb{C}^m$ as a complex Hilbert space.
\end{theorem}
We sketch the construction in Th. \ref{theoremrepresentation} with the aim to see how it can be used in specific cases, as follows. We introduce the space $\mathbb{R}^n$, and its canonical basis $h_1=(1, \ldots, 0, \ldots, 0)$, $h_2=(0, 1, 0, \ldots, 0)$, \ldots, $h_j=(0, \ldots, 1, \ldots)$, \ldots, $h_n=(0, \ldots, 1)$. The situation of the measurement $e$ and state $p$ can be represented by the vector
\begin{equation}
v(e,p)=\sum_{j=1}^n\mu(x_j,e,p)h_j
\end{equation}
which is a point of the simplex $S_n(e)$, the convex closure of the canonical basis $\{h_1, \ldots, h_j, \ldots, h_n\}$ in $\mathbb{R}^n$. We call $A_j(e,p)$ the convex closure of the vectors $\{h_1, h_2, \ldots$ $, h_{j-1}, v(e,p), h_{j+1}, \ldots, h_n\}$. We use this configuration to construct a micro-dynamical model for the measurement dynamics of $e$ for the entity in state $p$. This micro-dynamics is defined as follows, a vector $\lambda$ contained in the simplex $S_n(e)$, hence we have
\begin{equation}
\lambda=\sum_{j=1}^n\lambda_jh_j \quad 0 \le \lambda_j \le 1 \quad \sum_{j=1}^n\lambda_j=1
\end{equation}
determines the dynamics of the measurement $e$ on the state $p$ in the following way. If $\lambda \in A_j(e,p)$, and is not one of the boundary points (hence $\lambda$ is contained in the interior of $A_j(e,p)$), then the measurement $e$ gives with certainty, hence deterministically, rise to the outcome $x_j$, with the entity being in state $p$. If $\lambda$ is a point of the boundary of $A_j(e,p)$, then the outcome of the experiment $e$, the entity being in state $p$, is not determined. The probabilities $\mu(x_j, e, p)$ can then be derived from Lebesgue measuring the sets of relevant real numbers as subsets of $S_n(e)$. Indeed, as we have formulated the micro-dynamics of the measurement process $e$ for $S$ being in state $p$, we have that the $\mu(x_j, e, p)$, being the probability to obtain outcome $x_j$, is given by the Lebesgue measure of the set of vectors $\lambda$ that are such that this outcome is obtained deterministically, hence this are the $\lambda$ contained in $A_j(e,p)$, divided by the Lebesgue measure of the total set of vectors $\lambda$, which are the $\lambda$ contained in $S_n(e)$. This means that
\begin{equation} \label{probabilitylebesgue}
\mu(x_j, e, p)={m(A_j(e,p)) \over m(S_n(e))} .
\end{equation}
The right hand of \ref{probabilitylebesgue} can be evaluated as in \cite{aertssozzo2012a,aerts1986}.

Let us now come to the quantum representation. We introduce a set of orthogonal projection operators $\{Mk\ \vert k= 1,\ldots,n\}$ on a complex Hilbert $\mathbb{C}^m$ space, with $n \le m \le n^2$, that form a spectral family. This means that $M_k \perp M_l$ for $k\not=l$ and $\sum_{k=1}^nM_k=\mathbbmss{1}$, and we take the $M_k$ such that they are diagonal matrices in $\mathbb{C}^m$. More concretely, each $M_k$ is a matrix with 1's at some of the diagonal places, and zero's everywhere else. The number of 1's is between 1 and $n$, for each $M_k$, and the collections of 1's hang together, their mutual intersections being empty, and the union of all of them being equal to the collection of 1's of the unit matrix $\mathbbmss{1}$. The state is represented by a vector $w(e,p)$ of $\mathbb{C}^m$, such that
\begin{equation} \label{quantumsolution}
\mu(x_k, e, p)=\langle w(e,p)\ \vert M_k\ \vert w(e,p)\rangle
\end{equation}   
A possible solution is 
\begin{equation}
w(e,p)=\sum_{j=1}^m a_je^{i\alpha(e,p)_j}h_j \quad{\rm with} \quad a_j={1 \over b}\sqrt{\mu(x_j, e, p)}
\end{equation}
where $h_j$ is the canonical basis of $\mathbb{C}^m$, and $b$ is the dimension of the projector $M_k$ if $h_j$ is such that $M_kh_j=h_j$. But this is not the only solution, and it might also not be the appropriate solution for the situation we want to model. It shows however that a solution exists, which proves that it is always possible to built this local quantum model.

The above theorem is an application of the \emph{hidden measurement approach} that we elaborated in our Brussels research group in the eighties and nineties of the foregoing century, with the aim of formulating a contextual hidden variable model for quantum theory \cite{aerts1986,aerts1993,aerts1994,aerts1995,aertsaerts1997,aertsaertscoeckedhooghedurtvalckenborgh1997,aertsaertsdurtleveque1999,aertssven2002,aertssven2005,aertssven2006}.

With the above theorem we have constructed a representation of the collection of states and experiments that lead to the same set of outcomes. In this sense, the $\mathbb{R}^n$ model and the $\mathbb{C}^m$ that we have constructed is a model for the interaction between state and experiment. The set of outcomes constitutes a context in which this interaction takes place. In the next section we investigate in detail the examples to show the relevance of our representation theorem for the modeling of entanglement.

\vspace*{-0.2cm}
\section{Entanglement in QMod\label{entanglement}}
The representation theorem of QMod stated in Sec. \ref{sectionrepresentation} can be applied to specific entities and situations to show that entanglement, hence quantum structures, appear if suitable conditions are satisfied.

\subsection{Entanglement of two concepts\label{ent_conc}}
The first example that we take into account is a combination of two concepts. Let us consider the example of the entity which is the concept {\it Animal}, and let $e$ be a measurement where a person is asked to choose between the animal being a {\it Horse} or a {\it Bear}, hence $e$ is associated with two outcomes $\{H,B\}$. We consider only one state for {\it Animal}, namely the ground state which is the state where animal is just animal, i.e. the bare concept, and let us denote it by $p$. Let us denote by $\mu(H,e,p)$ the probability that {\it Horse} is chosen when $e$ is performed, and by $\mu(B,e,p)$ the probability that {\it Bear} is chosen in the same measurement. The following mathematical construction can now be elaborated. 

For the measurement $e$ we consider the vector space $\mathbb{R}^2$ and its canonical basis $\{(1, 0), (0, 1)\}$. The state $p$ is contextually represented with respects to the measurement $e$ by the vector $v(e,p)=(\mu(H,e,p),\mu(B,e,p))$ in $\mathbb{R}^2$. We introduce the vector $\lambda=(r, 1-r)$, with $0 \le r \le 1$, such that for $(r, 1-r)$ contained in the convex closure of $(1, 0)$ and $(\mu(H,e,p),\mu(B,e,p))$, we get outcome {\it Bear}, while for $(r, 1-r)$ contained in the convex closure of $(\mu(H,e,p),\mu(B,e,p))$ and $(0, 1)$ we get {\it Horse}. Let us calculate the respective lengths and see that we find back the correct probabilities. Denoting the length of the piece of line from $(1, 0)$ to $(\mu(H,e,p),\mu(B,e,p))$ by $d$, we have ${d \over \sqrt{2}}=\mu(B,e,p)$, and ${\sqrt{2}-d \over \sqrt{2}}=\mu(H,e,p)$. 

We can also construct a quantum mathematics model in $\mathbb{C}^2$. Therefore we consider the vector $w(e,p)=(\sqrt{\mu(H,e,p)}e^{i\alpha(e,p)_H},\sqrt{\mu(B,e,p)}e^{i\alpha(e,p)_B})$ in $\mathbb{C}^2$. We have $\mu(H,e,p)=|\langle (1,0)| w(e,p)\rangle|^2$ and $\mu(B,e,p)=|\langle (0,1)| w(e,p)\rangle|^2$, which shows that also the $\mathbb{C}^2$ construction gives rise to the correct probabilities. 

Now, we want to introduce explicitly the data that we collected in an experiment that we performed on test subjects and that is described in detail in \cite{aertssozzo2011}. Of the 81 persons that we asked to choose between {\it Horse} and {\it Bear} as good exemplars of the concept {\it Animal}, 43 chose for {\it Horse}, and 38 for {\it Bear}. Calculating the relative frequencies gives rise to probabilities $\mu(H,e,p)=0.53$ and $\mu(B,e,p)=0.47$. Hence
\begin{equation}
v(e,p)=(0.53,0.47) \quad w(e,p)=(0.73\ e^{i\alpha(e,p)_H}, 0.68\ e^{i\alpha(e,p)_B})
\end{equation}
are the vectors that respectively represent the state of the concept {\it Animal} with respect to this measurement and these data respectively in $\mathbb{R}^2$ and in $\mathbb{C}^2$.

We consider now the entity which is the concept {\it Acts}, where \emph{Acts} denotes here the action of emitting a sound, and the measurement $f$, where a person is invited to choose between {\it Growls} or {\it Whinnies}. Hence we have two outcomes $\{G,W\}$. Also for the concept {\it Acts} we consider only one state, the ground state, which we denote by $q$. The probabilities $\mu(G,f,q)$ and $\mu(W,f,q)$ are respectively the probability that {\it Growls} is chosen when $f$ is performed, and the probability that {\it Whinnies} is chosen in the same experiment. We again make the construction in $\mathbb{R}^2$ and $\mathbb{C}^2$ for the respective probabilities, giving rise to the vectors $v(f,q)=(\mu(G,f,q),\mu(W,f,q))$ and $w(f,q)=(\sqrt{\mu(G,f,q)}e^{i\alpha(f,q)_G},\sqrt{\mu(W,f,q)}e^{i\alpha(f,q)_W})$. The respective constructions allow one to reproduce the correct probabilities also in this case.  

Turning again to the data collected in the experiment described in \cite{aertssozzo2011}, of the 81 persons there were 39 choosing {\it Growls} and 42 choosing {\it Whinnies}. This leads to $\mu(G,f,q)=0.48$ and $\mu(W,f,q)=0.52$. Hence the vectors
\begin{equation}
v(f,q)=(0.48,0.52) \quad w(f,q)=(0.69\ e^{i\alpha(f,q)_G}, 0.72\ e^{i\alpha(f,q)_W})
\end{equation}
are the vectors that respectively represent the state of the concept {\it Acts} with respect to this measurement and the collected data in $\mathbb{R}^2$ and in $\mathbb{C}^2$, respectively.

We consider now the combination of both entities, hence the conceptual combination {\it The Animal Acts}, and again only one state, namely its ground state, which we denote $r$. Let $g$ be an experiment with four possible outcomes, namely {\it Horse} and {\it Growls} are chosen, {\it Horse} and {\it Whinnies} are chosen, {\it Bear} and {\it Growls} are chosen, or {\it Bear} and {\it Whinnies} are chosen. The set of possible outcomes is then $\{HG, HW, BG, BW\}$, and the corresponding probabilities are $\mu(HG,g,r)$, $\mu(HW,g,r)$, $\mu(BG,g,r)$ and $\mu(BW,g,r)$. 

If we develop the mathematical construction explained in our representation theorem, we need to consider $\mathbb{R}^4$, and $\mathbb{C}^4$ and the corresponding simplex in $\mathbb{R}^4$. This is the crucial aspect that makes it possible to model entanglement, as our analysis will show.

We firstly recall that $\mathbb{R}^4$ is isomorphic to $\mathbb{R}^2 \otimes \mathbb{R}^2$, and $\mathbb{C}^4$ is isomorphic to $\mathbb{C}^2 \otimes \mathbb{C}^2$, and it are these isomorphisms that allow the modeling of entanglement in a straightforward way. The canonical basis of $\mathbb{R}^2 \otimes \mathbb{R}^2$ and of $\mathbb{C}^2 \otimes \mathbb{C}^2$ is
\begin{equation}
h_1=(1,0)\otimes(1,0) \quad h_2=(1,0)\otimes(0,1) \quad h_3=(0,1)\otimes(1,0) \quad h_4=(0,1)\otimes(0,1)
\end{equation}
Hence, we have
\begin{eqnarray} \label{entangledvectorreal}
&v(g,r)=\mu(HG,g,r)h_1+\mu(HW,g,r)h_2+\mu(BG,g,r)h_3+\mu(BW,g,r)h_4 \\ \label{entangledvectorcomplex}
&w(g,r)=\sqrt{\mu(HG,g,r)}e^{\alpha(g,r)_{HG}}h_1+\sqrt{\mu(HW,g,r)}e^{\alpha(g,r)_{HW}}h_2 \nonumber \\
&+\sqrt{\mu(BG,g,r)}e^{\alpha(g,r)_{BG}}h_3+\sqrt{\mu(BW,g,r)}e^{\alpha(g,r)_{BW}}h_4
\end{eqnarray}
and can prove the following theorem.
\begin{theorem}
$v(g,r)$ equals the product state $v(e,p)\otimes v(f,q)$ (and then also $w(g,r)$ equals the product state $w(e,p)\otimes w(f,q)$) iff the probabilities satisfy
\begin{eqnarray} \label{productprobabilities01}
&\mu(HG,g,r)=\mu(H,e,p)\mu(G,f,q) \quad \mu(HW,g,r)=\mu(H,e,p)\mu(W,f,q) \\ \label{productprobabilities02}
&\mu(BG,g,r)=\mu(B,e,p)\mu(G,f,q) \quad \mu(BW,g,r)=\mu(B,e,p)\mu(W,f,q)
\end{eqnarray}
\end{theorem}
\begin{proof} We have
\begin{eqnarray}
&v(e,p)\otimes v(f,q)=(\mu(H,e,p),\mu(B,e,p))\otimes(\mu(G,g,q),\mu(W,g,q)) \nonumber \\
&=\mu(H,e,p)\mu(G,g,q)h_1+\mu(H,e,p)\mu(W,g,q)h_2+\mu(B,e,p)\mu(G,g,q)h_3+\mu(B,e,p)\mu(W,g,q)h_4
\end{eqnarray}
Analogously, we have
\begin{eqnarray}
&w(e,p)\otimes w(f,q)=\sqrt{\mu(H,e,p)\mu(G,f,q)}e^{\alpha(e,p)_{H}}e^{\alpha(f,q)_{G}}h_1+\sqrt{\mu(H,e,p)\mu(W,f,q)}e^{\alpha(e,p)_{H}}e^{\alpha(f,q)_{W}}h_2 \nonumber \\
&+\sqrt{\mu(B,e,p)\mu(G,f,q)}e^{\alpha(e,p)_{B}}e^{\alpha(f,q)_{G}}h_3+\sqrt{\mu(B,e,p)\mu(W,f,q)}e^{\alpha(e,p)_{B}}e^{\alpha(f,q)_{W}}h_4 . 
\end{eqnarray}
\end{proof}
Let us now consider the data that we collected in the experiment described in \cite{aertssozzo2011}, and see that we encountered there an entangled state. From the 81 persons that participated in the experiment, there were 4 persons that choose {\it The Horse Growls}, 51 persons that choose {\it The Horse Whinnies}, 21 persons that choose {\it The Bear Growls}, and 5 persons that choose {\it The Bear Whinnies}. This leads to probabilities $\mu(HG,g,r)=0.05$, $\mu(HW,g,r)=0.63$, $\mu(BG,g,r)=0.26$ and $\mu(BW,g,r)=0.06$. This means that
\begin{eqnarray} \label{entangledvectorcomplex2}
&v(g,r)=0.05\ h_1+0.63\ h_2+0.26\ h_3+0.06\ h_4 \\ \label{entangledvectorcomplex3}
&w(g,r)=0.22\ e^{\alpha(g,r)_{HG}}h_1+ 0.79\ e^{\alpha(g,r)_{HW}}h_2+0.51\ e^{\alpha(g,r)_{BG}}h_3+0.25\ e^{\alpha(g,r)_{BW}}h_4 
\end{eqnarray} 
are the vectors that represent the state of the concept {\it The Animal Acts} with respect to this measurement and the collected data in $\mathbb{R}^4$ and $\mathbb{C}^4$, respectively. It is easy to check that the vectors in (\ref{entangledvectorcomplex2}) and (\ref{entangledvectorcomplex3}) represent a state that is not a product state in the sense that the probabilities corresponding to the joint measurement are not equal to the products of the probabilities corresponding to the component measurements. What is however much more conclusive with respect to the state of {\it The Animal Acts} being a state of entanglement, is that it can be proven that no component probabilities can possibly exist that give rise to the experimental values measured for the joint probabilities. This result is stated by means of the following theorem.

\begin{theorem} There do not exist numbers $a_1, a_2, b_1, b_2$ contained in the interval $[0,1]$, such that $a_1+a_2=1$, and $b_1+b_2=1$, and such that $a_1b_1=0.05$, $a_2b_1=0.63$, $a_1b_2=0.26$ and $a_2b_2=0.06$.
\end{theorem}
\begin{proof} Let us suppose that such numbers do exist. From $a_2b_1=0.63$ follows that $(1-a_1)b_1=0.63$, and hence $a_1b_1=1-0.63=0.37$. This is in contradiction with $a_1b_1=0.05$.
\end{proof}
It is important to observe that in case we do not have the equalities (\ref{productprobabilities01}) and (\ref{productprobabilities02}) for the probabilities satisfied, and hence are in a situation of entanglement, we can model this within the $\mathbb{R}^2 \otimes \mathbb{R}^2$ tensor product space, and also in the $\mathbb{C}^2 \otimes \mathbb{C}^2$ tensor product space. It is just that in this case the vectors $v(g,r)$ and $w(g,r)$ will not be product vectors, but entangled vectors, i.e. the sum of product vectors, as can be seen in (\ref{entangledvectorreal}) and (\ref{entangledvectorcomplex2}). We also recall that we do not need any linear structure at all for the global set of states $\Sigma$, it is only the representation of this set of states due to the representation theorem \ref{theoremrepresentation} presented in section \ref{sectionrepresentation}, and proven in \cite{aertssozzo2012a}, which is a space of real numbers or contains a linear structure as a complex space. But, what is most important of all to recall is that this `local contextual real-space or complex-linear structure' can always be realized independent of the entity and situation considered. The analogy with how general relativity has been mathematically constructed as a generalization of special relativity can now be very well illustrated. Indeed, the real-space or linear structure is only local, for a fixed set of outcomes. Therefore, the formalism we propose is a generalization of standard quantum mechanics in the sense that, when the real space representation is used, no linearity at all is involved, and when the complex space representation is used, linearity is present only locally. Moreover, even in the latter representation, it is not necessarily the case that also globally the set of states can be made into a linear vector space. Only when this can be done, hence when all the local linearities join into one global linearity, the formalism we propose reduces to the standard quantum theoretical formalism. Another way of expressing the above is that Quantum Model Theory is realized by means of a `contextual linear formalism'.

\subsection{Entanglement of General Entities}
The real and complex representations of a state of a compound entity in terms of the corresponding representations of the states of the component entities that we have constructed in Sec. \ref{ent_conc} for two concepts can be extended to two general entities. In the following theorem we make this construction and indicate how entangled states can be identified.

\begin{theorem}
Entangled states can be identified for general compound entities modeled in QMod 
\end{theorem}
\begin{proof}
Let $S$ and $T$ be two entities in the states $p$ and $q$, respectively, and let the measurements $e$ and $f$ be performed on $S$ and $T$, respectively. Suppose that $\{x_1,\ldots, x_n\}$ is the set of outcomes of $e$ and $\{y_1,\ldots y_n\}$ is the set of outcomes of $f$, and denote by $\mu(x_j,e,p)$, $\mu(y_k,f,q)$ the corresponding probabilities. Finally, let 
\begin{eqnarray}
v(e,p)=(\mu(x_1,e,p),\ldots, \mu(x_n,e,p)) \\ v(f,q)=(\mu(y_1,f,q),\ldots,\mu(y_n,f,q))
\end{eqnarray}
\begin{eqnarray}
w(e,p)=(\sqrt{\mu(x_1,e,p)}e^{\alpha(e,p)_1},\ldots, \sqrt{\mu(x_n,e,p)}e^{\alpha(e,p)_n}) \\ w(f,q)=(\sqrt{\mu(y_1,f,q)}e^{\alpha(f,q)_1},\ldots, \sqrt{\mu(y_n,f,q)}e^{\alpha(f,q)_n})
\end{eqnarray}
be the contextual representations of $(e,p)$ and $(f,q)$ in ${\mathbb R}^{n}$ and ${\mathbb C}^{n}$, respectively. Finally, let $U$ be the compound entity made up of $S$ and $T$, in the state $r$. Let the measurement of $g$ on $U$ consisting of a measurement of $e$ on $S$ and $f$ on $T$ so that the set of possible outcomes of $g$ is $\{(x_1,y_1),\ldots, (x_j,y_k)\ldots, (x_n,y_n)\}$, and the set of corresponding probabilities $\{\mu((x_j,y_k),g,r)\}$. By repeating the procedure of Sec. \ref{ent_conc}, we can write
\begin{eqnarray} \label{entangledvectorreal_gen}
&v(g,r)=\sum_{j,k}\mu((x_j,y_k),g,r)h_{jk} \\ \label{entangledvectorcomplex_gen}
&w(g,r)=\sum_{jk}\sqrt{\mu((x_j,y_k),g,r)}e^{\alpha(g,r)_{jk}}h_{jk}
\end{eqnarray}
where $\{h_{kj}\ \vert k,j\in\{1,\ldots,n\}\}$ is the canonical base of ${\mathbb R}^{n}\otimes{\mathbb R}^{n}$, which is a $n^2$ dimensional real space, hence isomorphic to ${\mathbb R}^{n^2}$. Moreover, reasoning as in Th. 2, we get that $v(g,r)=v(e,p)\otimes v(f,q)$ and $w(g,r)=w(e,p)\otimes w(f,q)$) iff the probabilities satisfy
\begin{equation} \label{productprobabilities}
\mu((x_j,y_k),g,r)=\mu(x_j,e,p)\mu(y_k,f,q)
\end{equation}
In case (\ref{productprobabilities}) is not satisfied, $r$ is an entangled state.
\end{proof}

\subsection{Entanglement of two vessels of water\label{ent_vessels}}
Let us come to the second example. We consider two vessels of water, each containing a volume of water, between 0 and 20 liters. Suppose that we are in a situation where we lack knowledge about the exact volume contained in each vessel. We call the state of the left vessel $p$ and the state of the right vessel $q$. We consider measurements $e$ and $f$ for the left and right vessel respectively, that consist in pouring out the water by means of a siphon, collecting it in reference vessels, where we can read of the volume of collected water. We attribute outcome $M$ if the volume is more than 10 liters and the outcome $L$ if it is less than 10 liters. We introduce the probabilities $\mu(M,e,p)$ and $\mu(L,e,p)$ for the outcomes $M$ and $L$ of $e$ on the left vessel, and the probabilities $\mu(M,f,q)$ and $\mu(L,f,q)$ for the outcomes $M$ and $L$ of $f$ on the right vessel. 

We then consider the joint entity consisting of the two vessels of water and denote the state of this joint entity by $r$. The measurement $g$ consists in pouring out the water of the left vessel with the siphon, and also of the right vessel, with another siphon. Volumes of water are collected at left and at right in two reference vessels, and four outcomes are considered $\{MM, LM, ML, LL\}$. The outcome $MM$ corresponds to left as well as right vessel giving rise to the collection of more than 10 liters, and outcome $LL$ corresponds to left as well as right vessel giving rise to the collection of less than 10 liters. The other two outcomes $ML$ ($LM$) correspond to the left vessel giving rise to more (less) than 10 liters and right vessel giving rise to less (more) than 10 liters. The probabilities $\{\mu(MM,g,r), \mu(LM,g,r), \mu(ML,g,r), \mu(LL,g,r)\}$ correspond to these four outcomes. Obviously, if nothing extra happens between the two vessels, the joint probabilities will be product probabilities, which means that we have
\begin{eqnarray}
&&\mu(MM,g,r)=\mu(M,e,p)\mu(M,f,q) \quad \mu(LM,g,r)=\mu(L,e,p)\mu(M,f,q) \\
&&\mu(ML,g,r)=\mu(M,e,p)\mu(L,f,q) \quad \mu(LL,g,r)=\mu(L,e,p)\mu(L,f,q)
\end{eqnarray}
This shows that there is no entanglement, and that in the local contextual model in $\mathbb{R}^2 \otimes \mathbb{R}^2$ and $\mathbb{C}^2 \otimes \mathbb{C}^2$, we can represent the state $r$ by means of product states $v(e,p)\otimes v(f,q)$ and $w(e,p)\otimes w(f,q)$. 

Let us propose a situation which is more concretely defined, and allows us to derive some numerical values for the probabilities. Thus, we suppose that, for each vessel, the lack of knowledge about the volume of the water contained in the vessel, is evenly distributed. As a consequence of this extra hypothesis, the numerical values for all the probabilities are determined from reasons of symmetry, and we have
\begin{eqnarray}
&&\mu(M,e,p)=\mu(L,e,p)=\mu(M,f,q)=\mu(L,f,q)={1 \over 2} \\ \label{values01}
&&\mu(MM,g,r)=\mu(LM,g,r)=\mu(ML,g,r)=\mu(LL,g,r)={1 \over 4} .
\end{eqnarray} 
We want to consider now another state of the two vessels, and show that this new state is entangled. It is a state where we connect the two vessels of water by a tube, such that they form `connected vessels of water', and we put exactly 20 liters of water in the whole of the connected vessels. Let us denote this state by $s$. Knowing that the measuring of the volume of each vessel consist of pouring out the water by a siphon, for the state $s$, we find that the volume of both vessels, i.e. the water being collected by the siphons, is strictly correlated. Indeed, if we find less than 10 liters in the left vessel, we find more than 10 liters in the right vessel, and vice versa. This means that we never get outcome $MM$ and $LL$, and hence we have $0=\mu(MM,g,s)=\mu(LL,g,s)$, while $1=\mu(ML,g,s)+\mu(LM,g,s)$. Let us investigate whether $s$ is an entangled state. To this aim, we suppose that $s$ is a product state, and see what follows from this hypothesis. If $s$ is a product state we have
\begin{eqnarray}
0=\mu(M,e,p)\mu(M,f,q)=\mu(L,e,p)\mu(L,f,q)
\end{eqnarray}
which implies that $\mu(M,e,p)$ and $\mu(L,f,q)=0$ or $\mu(M,f,q)$ and $\mu(L,e,p)=0$. Hence, this means that the left vessel contains with certainty less than 10 liters, and the right vessel contains with certainty more than 10 liters, or vice versa. Suppose we have $\mu(M,e,p)$ and $\mu(L,f,q)=0$. Then $\mu(L,e,p)=1$ and $\mu(M,f,q)=1$, but hence $\mu(LM,g,r)=\mu(L,e,p)\mu(M,f,q)=1$ and $\mu(ML,g,r)=0$. This is only possible if the siphon of the right vessel would pour out no water at all, and all the water would be poured out by the siphon of the left vessel. This is very improbably, not to say impossible, and hence in case of a realistic situation we have both $\mu(LM,g,r)$ and $\mu(ML,g,r)$ different from zero, which means that $s$ is an entangled state.

Let us again introduce an extra hypothesis that will allow us to derive numerical values for the probabilities in the state $s$, and prove that $s$ is entangled. Thus, we suppose that both siphons are chosen at random to be applied to the left or to the right, and also all other parameters involved in applying the siphons are chosen at random, e.g. the starting time of siphoning is at random. In this case, we have probability one half that the left siphon will pour out more than 10 liters -- and in this case the right siphon pours out less than 10 liters -- and probability one half that the right siphon will pour out more than 10 liters of water -- and in this case the left siphon pours out less than 10 liters. This means that
\begin{equation} \label{values02}
\mu(ML,g,s)=\mu(LM,g,s)={1 \over 2} .
\end{equation}
If we compare (\ref{values01}) with (\ref{values02}), we see that if the extra hypothesis is satisfied, the state $s$ is not a product state. Hence $s$ is an entangled state. Again, like in the case of the example {\it The Animal Acts}, we can show that no component probabilities can exist to give rise to these joint probabilities.
\begin{theorem} There do not exist numbers $a_1, a_2, b_1, b_2$ contained in the interval $[0,1]$, such that $a_1+a_2=1$, and $b_1+b_2=1$, and such that $a_1b_1=0$, $a_2b_1=0.5$, $a_1b_2=0.5$ and $a_2b_2=0$.
\end{theorem}
\begin{proof} Let us suppose that such numbers do exist. From $a_2b_1=0.5$ follows that $(1-a_1)b_1=0.5$, and hence $a_1b_1=1-0.5=0.5$. This is in contradiction with $a_1b_1=0$.
\end{proof}
The entangled states that we identity in the way shown above do not contain already the best known characteristic of entanglement, namely the violation of Bell-type inequalities. The reason for this is that locally, hence if only one measurement context is considered, Bell-type inequalities cannot even be defined. Different measurement contexts need to be confronted with each other to come to an investigation of the violation of Bell-type inequalities. In \cite{aertssozzo2011} we show that for the data with respect to the combination of concepts {\it Animal} and {\it Acts}, in effect, also Bell-type inequalities are violated in case more measurement contexts are considered for this entity {\it The Animal Acts}. That the vessel of water example also violates Bell-type inequalities of more measurement contexts are considered was shown by one of the authors in earlier work \cite{aerts1982,aerts1991}. In forthcoming work we will show how the consideration of different contexts on QMod allows the identification of compatibility and non compatibility, again without the necessity of linearity. It will also be proven that the violation of Bell-type inequalities is due to the presence of both aspects entanglement and non compatibility.

\end{document}